\documentclass[sigconf, authorversion]{acmart}
\usepackage{amsmath}
\usepackage{hyperref}
\usepackage{siunitx}
\usepackage{graphicx}
\usepackage{subcaption} %

\newcommand{\Rt}{\ensuremath{\tilde{R}}}
\newcommand{\R}{\ensuremath{R}}

\newcommand{\infint}{\ensuremath{\int_{-\infty}^{\infty}}}

\newif\ifmarkrevision

\markrevisionfalse
\newcommand{\blueon}{%
  \ifmarkrevision
    \color{blue}%
  \fi
}

\newcommand{\revlabel}[1]{%
  \ifmarkrevision
    \textcolor{blue}{\textbf{[Rev\##1]: }}%
  \fi
}

\newcommand{\blueoff}{\color{black}}

\AtBeginDocument{%
  }

\setcopyright{acmlicensed}

\copyrightyear{2026}
\acmYear{2026}
\setcopyright{cc}
\setcctype{by}
\acmConference[WiSec '26]{Proceedings of the 19th ACM Conference on Security and Privacy in Wireless and Mobile Networks}{June 30-July 03, 2026}{Saarbrücken, Germany}
\acmBooktitle{Proceedings of the 19th ACM Conference on Security and Privacy in Wireless and Mobile Networks (WiSec '26), June 30-July 03, 2026, Saarbrücken, Germany}
\acmDOI{10.1145/3765613.3811683}
\acmISBN{979-8-4007-2201-1/2026/06}

\begin{document}

\title{Security Analysis of Time-of-Arrival Estimation via Cross-Correlation under Narrow-Band Conditions
}
\author{Claudio Anliker}
\email{claudio.anliker@inf.ethz.ch}
\orcid{TODO}
\affiliation{%
  \institution{ETH Zurich}
  \city{Zurich}
  \country{Switzerland}
}

\author{Daniele Coppola}
\orcid{TODO}
\affiliation{%
  \institution{ETH Zurich}
  \city{Zurich}
  \country{Switzerland}
}
\email{daniele.coppola@inf.ethz.ch}

\author{Giovanni Camurati}
\orcid{TODO}
\affiliation{%
  \institution{ETH Zurich}
  \city{Zurich}
  \country{Switzerland}
}
\email{giovanni.camurati@inf.ethz.ch}

\author{Srdjan \v{C}apkun}
\orcid{TODO}
\affiliation{%
  \institution{ETH Zurich}
  \city{Zurich}
  \country{Switzerland}
}
\email{srdjan.capkun@inf.ethz.ch}

\renewcommand{\shortauthors}{Claudio Anliker, Daniele Coppola, Giovanni Camurati, and Srdjan Čapkun} 
\begin{abstract}
Time-of-arrival (ToA) estimation via cross-correlation is an essential building block of time-of-flight ranging. However, in narrowband systems, it is notoriously difficult to protect against \emph{distance-decreasing attacks} such as Early-Detect/Late-Commit (ED/LC).

We present and analyze two new attacks that reshape ranging signals to compromise correlation-based ToA estimation. The first attack multiplies the signal by a symbol-periodic waveform in the time domain, while the second passes it through a negative group delay (NGD) filter. In contrast to ED/LC, our attacks do not require real-time symbol detection or adaptive compensation; they are completely \emph{symbol-agnostic}. We describe implementation strategies for both attacks and discuss NGD filtering in the context of Bluetooth Channel Sounding (CS), a recent narrowband ranging system. To this end, we simulate an NGD circuit in LTspice and a ToA estimator in MATLAB, demonstrating that the attack can result in distance reductions of up to \SI{18}{\meter} against Bluetooth CS RTT ranging. Finally, we verify the feasibility of the NGD approach by building a prototype using commercial off-the-shelf components.
\end{abstract}



\maketitle

\section{Introduction}
\emph{Secure ranging} refers to the capability to upper-bound the physical distance between devices~\cite{luo2024secure}. It is at the core of proximity-based access control systems, enabling the use of consumer devices, such as smartphones or smartwatches, as physical-world keys. Many such systems deploy time-of-flight (ToF) ranging~\cite{IEEE802.15.4z, IEEE802.15.4ab, IEEE802.11az, BluetoothSpec}, which estimates the distance between devices from the propagation delay of wireless signals. 

ToF ranging relies on accurate time-of-arrival (ToA) estimation. A common method involves cross-correlating the received signal, which carries a pre-shared waveform, with a local template. In an additive white Gaussian noise (AWGN) channel, this template acts as a matched filter, and the correlation peak yields the maximum-likelihood estimate of the ToA.

A \emph{distance-decreasing attack} or \emph{distance-reduction attack} attempts to shorten the distance estimate produced by the transceivers. In ToF ranging, these attacks typically aim to bias the ToA estimates used to derive the ToF and, thus, the distance.
To mitigate distance-decreasing attacks, ranging signals must be unpredictable; otherwise, a synchronized attacker could manipulate the ToA by injecting the signal ahead of schedule. Furthermore, any real-time information leakage must not enable a ToA advance exceeding the system's security margin. Therefore, ToA estimation using waveforms with high temporal redundancy is risky; an attacker could identify a symbol after only \emph{partial reception} and inject a crafted waveform to bias the receiver's ToA estimate. This is typical of narrowband systems, which tend to require longer symbol durations due to the time-frequency tradeoff. 

Clulow et al. first discussed the problem of long symbols in the context of ranging in 2006~\cite{DBLP:conf/esas/ClulowHKM06}. Their work formed the basis for attacks against various ranging technologies, such as RFID/433MHz ASK/FSK~\cite{DBLP:conf/wisec/HanckeK08}, Impulse-Radio Ultra-Wideband~\cite{DBLP:conf/wisec/FluryPPHB10}, Chirp Spread Spectrum ranging~\cite{DBLP:conf/wisec/RanganathanDFC12}, multi-carrier (OFDM) ranging~\cite{leu2021security}, and GNSS systems~\cite{zhang2019effects}. Current ranging technologies deploy different strategies to resolve this issue: Ultra-Wideband Ranging (IEEE 802.15.4z/ab) uses only nanosecond-long symbols~\cite{IEEE802.15.4z, IEEE802.15.4ab}, Next Generation Positioning (IEEE 802.11az) increases signal entropy to render long symbols unpredictable~\cite{Batra2020_IEEE80211az}, and Bluetooth Channel Sounding relies on detection mechanisms in the receiver~\cite{BluetoothSpec}\footnote{Vol.6, Part H, Section 3.5.3.\label{bl_countermeasures}}.

The idea of detecting symbols after a partial observation and reactively manipulating the waveform is called \emph{Early-Detect/Late-Commit} (ED/LC). While ED/LC is conceptually straightforward, its adaptive nature presents practical challenges: the detection phase typically requires a considerable signal-to-noise ratio (SNR) advantage over the legitimate receiver, the manipulation may be detected by countermeasures~\cite{BluetoothSpec}\footref{bl_countermeasures}, and the execution is highly time-constrained.

\paragraph{Our Contributions:}
In this work, we extend the study of ED/LC by exploring how symbol-synchronous signal reshaping enables distance-decreasing attacks against correlation-based ranging systems. Unlike ED/LC, our attacks are \emph{symbol-agnostic}; they bypass the need for real-time symbol decoding, considerably lowering the practical barriers to exploitation. In summary, our contributions are:

\begin{itemize}
    \item \emph{Novel distance-decreasing attacks:} \revlabel{2}\blueon We present two symbol-agnostic strategies: \emph{temporal masking} against linear modulations and \emph{NGD filtering}, a more general, modulation-agnostic approach utilizing negative group delay (NGD) filters.\blueoff
    \item \emph{Systematic Analysis:} We formalize how these attacks violate the core assumptions of correlation-based ToA estimation and derive the specific requirements for a successful distance reduction.
    \item \emph{Proof of concept:} We validate NGD filtering by simulating and prototyping an analog circuit using the LTC6268 operational amplifier, and outline the implementation of the temporal masking attack.
    \item \emph{Application to Bluetooth CS:} Using a MATLAB-based receiver, we evaluate NGD filtering against Bluetooth Channel Sounding (CS). Our results confirm distance reductions of up to \SI{18}{\meter} and demonstrate that the attack is hard to detect with security metrics defined in the Bluetooth Core Specification~\cite{BluetoothSpec}.
\end{itemize}

\section{Background}\label{sec:background}
\paragraph{ToF computation} Two-way ranging (TWR) is a widely used ToF protocol involving a bi-directional message exchange between two devices (see \autoref{fig:twr}). In Bluetooth CS, the \emph{initiator} starts the procedure and the \emph{reflector} responds. The initiator computes the ToF as:

\begin{equation}
   T_\mathrm{ToF} = \frac{T_{\mathrm{round}} - (1-\hat{e})T_{\mathrm{reply}}}{2},
\end{equation}
where $T_{\mathrm{reply}}$ is the reply time reported by the reflector and $\hat{e}$ the clock drift between the devices. The primary advantage of TWR is that it eliminates the need for tight clock synchronization.

\paragraph{ToA estimation} ToF ranging requires both devices to estimate the time of arrival (ToA) of received messages with high precision. This is typically achieved via cross-correlation: the transmitter sends a pre-shared signal, which the receiver correlates with a local reference template. The receiver identifies the ToA, $\tau_0$, by locating the correlation peak:

\begin{align}
\tau_0 = \underset{\tau}{\mathrm{argmax}}\left|\int y(t)x^{*}(t-\tau)\,dt\right|,
\end{align}
where $y(t)$ is the received signal, which is typically distorted by the channel, and $x(t)$ is the template. In an additive white Gaussian noise (AWGN) scenario, $\tau_0$ is the maximum-likelihood estimate for the ToA. In digital receivers, this integral is represented by a discrete sum, and the peak of the underlying continuous function lies between two samples. To achieve sub-sample ToA accuracy, receivers typically use polynomial interpolation, leveraging the fact that the correlation peak is approximately parabolic.

\paragraph{ToA estimation accuracy} The Cramér-Rao Lower Bound (CRLB) for the variance of any unbiased ToA estimator is defined as~\cite{gezici2005localization}:
\begin{equation}\label{eq:crlb}
\mathrm{Var}(\hat{\tau}) \geq \frac{1}{8\pi^2\cdot\mathrm{SNR} \cdot \beta_{\mathrm{RMS}}^2},
\end{equation}
where $\mathrm{SNR}$ represents the signal-to-noise ratio prior to matched filtering and $\beta_{\mathrm{RMS}}$ is the root-mean-square (RMS) bandwidth, defined as
\begin{equation}\beta_{\mathrm{RMS}} = \left(\frac{\int_{-\infty}^{\infty} f^2 |S(f)|^2 \, df}{\int_{-\infty}^{\infty} |S(f)|^2 \, df}\right)^{1/2}.\end{equation}

Intuitively, a larger $\beta_{\mathrm{RMS}}$ reduces the estimation error because signals occupying a wider spectrum exhibit sharper transitions in the time domain, resulting in a narrower autocorrelation peak that is easier to localize. To improve the SNR, ranging systems may employ higher transmission power, high-gain antennas, or longer integration times (i.e., longer signal sequences). The CRLB explains why Ultra-Wideband (UWB), with a bandwidth of $\approx\SI{500}{\mega\hertz}$, can achieve sub-decimeter ranging accuracy using a single measurement, while narrowband technologies like Bluetooth CS RTT (\SI{2}{\mega\hertz}) are inherently less precise at the same SNR.

\begin{figure}[t]
    \centering
        \resizebox{0.8\linewidth}{!}{
    \includegraphics{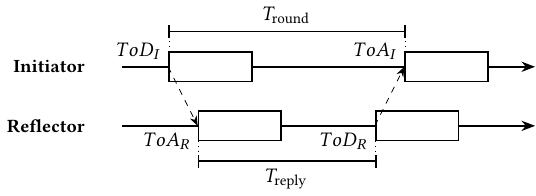}
    }
     \caption{The two-way ranging protocol. The reflector reports $T_{\mathrm{reply}}$ to the initiator, which computes the ToF.} 
    \label{fig:twr}
\end{figure}

\paragraph{Bandwidth and ToA integrity}
Narrowband systems are not only less accurate but also more difficult to secure. Reducing a signal's bandwidth stretches the waveform in the time domain, creating a window of opportunity for distance-decreasing attacks such as \emph{Early-Detect/Late-Commit (ED/LC)}. ED/LC is an umbrella term for attacks where an attacker exploits long symbol durations to infer a symbol from a partial observation (early detection) and subsequently injects an crafted signal to advance the receiver’s ToA estimate (late commit). Figure 4 of~\cite{leu2021security} illustrates this process.

\section{Attacking the ToA Estimation}\label{sec:attack_intro}
The ED/LC attack mentioned above is only relevant to ToF ranging systems for which a ToA advance of less than a symbol duration constitutes a security risk. This is typically the case for single-carrier, narrowband ranging. In this work, we focus on such systems.

In practice, performing ED/LC can be challenging: the early detection phase is time-critical and typically requires a high signal-to-noise ratio (SNR), and the late-commit stage may cause detectable mid-symbol transitions ~\cite{BluetoothSpec}\footnote{Vol.6, Part H, Section 3.5.3}. By contrast, the attacks we propose are \emph{symbol-agnostic}, demonstrating that an attacker can manipulate correlation-based ToA estimation without symbol decoding or adaptive compensation.

\subsection{Overview}
We first explain the intuition behind the attacks and provide a systematic analysis in \autoref{sec:analysis}.

\begin{figure}
    \centering
    \includegraphics[width=\linewidth]{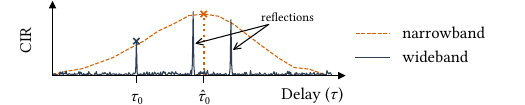}
     \caption{Narrowband signals typically blur the channel impulse response ($\mathrm{CIR}$) into a single peak, resulting in an inaccurate ToA estimate ($\hat{\tau}_0$). By contrast, a wideband system can estimate the first-path $\tau_0$ with high precision.}
    \label{fig:nb_wb_comparison}
\end{figure}

\subsubsection{The Crux with Narrowband ToA Estimation} 
Narrowband symbols are typically longer than the expected channel delay spread, which is in the hundreds of nanoseconds~\cite{molisch2009ultra}. Consequently, cross-correlation collapses the channel impulse response (CIR) into a single peak, which is dominated by the shape of the narrowband waveform (see \autoref{fig:nb_wb_comparison}). This has profound security implications: In UWB, advancing the correlation peak in a conservative configuration\footnote{In the sense that the leading-edge (first-path) detection favors security over reliability.} is infeasible without knowing the pseudo-random \emph{symbol values} beforehand. By contrast, in narrowband systems, manipulating the \emph{symbol shape} may be sufficient to considerably change the correlation peak, thus compromising the ToA estimate.

\begin{figure}
    \centering
    \includegraphics[width=\linewidth]{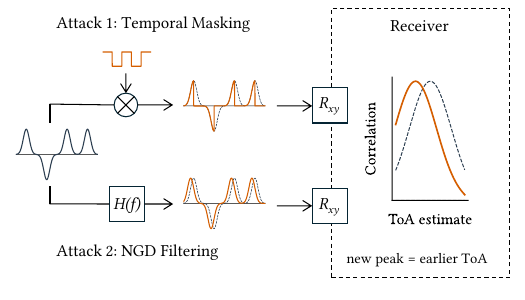}
     \caption{Overview of the proposed attacks: The ranging signal is multiplied by a mask (Attack 1) or passed through an NGD filter (Attack 2). In both cases, the output (orange/solid) advances the correlation peak. }
    \label{fig:overview}
\end{figure}

\subsubsection{Intuition}\label{sec:intuition}
The goal of the proposed attacks is to \emph{reshape} the ranging signal $x(t)$ over the air such that the receiver's correlation peak is shifted to the left by a non-negligible $\Delta t$. From the receiver's perspective, this shift indicates that the signal arrived $\Delta t$ earlier. 
As shown in \autoref{fig:equation_overview}, we achieve this by multiplying the ranging signal by a mask $m(t)$ or applying a filter $h(t)$. Both methods yield similar results, as illustrated in~\autoref{fig:overview}. 

Let $\tilde{x}(t)$ denote the attack signal. Ideally, an attacker would advance the correlation peak by transmitting $\tilde{x}(t) = x(t+\Delta t)$. However, this signal is not causal, as it requires a perfect prediction of $x(t)$. Instead, an attacker could \emph{approximate} the advanced signal $x(t+\Delta t)$ using a first-order Taylor expansion:

\begin{equation}
\tilde{x}(t) \approx x(t+\Delta t) = x(t) +\Delta tx'(t) + \mathcal{O}((\Delta t)^2), \label{eq:tildex_taylor} 
\end{equation}
which relies only on the current state of $x(t)$ and is, thus, causal. 

In the following, we formalize this requirement more rigorously (\autoref{sec:analysis}) and show how to construct $\tilde{x}(t)$ without symbol knowledge (\autoref{sec:attack_mixer} and \autoref{sec:attack_filter}).

\begin{figure}
\centering
\includegraphics[width=0.8\linewidth]{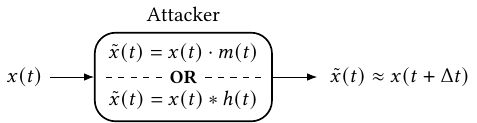}
\caption{Attack description with equations. The signal $x(t)$ is altered over the air, and the output $\tilde{x}(t)$ appears to the receiver like an advanced version of $x(t)$. }
\label{fig:equation_overview}
\end{figure}

\subsection{Analysis}\label{sec:analysis}
\revlabel{4}\blueon
In this section, we formally derive the requirement for a successful ToA advancement.
\subsubsection{Assumptions}
The proposed attacks are symbol-agnostic and affect each symbol of the ranging message equally. Without loss of generality, we can thus consider a message carrying a \emph{single, randomly drawn symbol} and rely on the linearity of the cross-correlation operator to generalize the result to longer messages. 

\paragraph{Continuous-time model:} To decouple the analysis from specific receiver designs or sampling rates, we model all signals and their correlations in the continuous-time domain. This approach reflects the operational reality of digital receivers, which typically leverage interpolation to characterize the underlying continuous waveform to achieve sub-sample ToA precision. This model implies a linear time-invariant (LTI) front-end, an ideal anti-aliasing filter before the sampling stage to prevent aliasing of high-frequency noise, and sufficient oversampling.

\paragraph{Attacker and channel:} We define the baseband signal at the receiver as
\begin{equation}
    \tilde{x}(t) = x(t) + \Delta x(t),\label{eq:x_tilde}
\end{equation}
where $x(t): \mathbb{R} \rightarrow \mathbb{C}$ is the baseband signal of the symbol transmitted by the legitimate sender and $\Delta x(t)$ the \emph{perturbation} introduced by an attacker. Both signals are smooth, continuous, and have finite support.

In a frequency-flat fading model, the input $x(t)$ and the output $y(t)$ are related by
\begin{equation}
    y(t) = hx(t) + n(t),
\end{equation}
where the channel $h$ is modeled using a single complex number, and $n(t)$ is AWGN. \autoref{eq:x_tilde} implies an ideal, noiseless channel with zero delay, i.e., $h = 1$ and $n(t) = 0$. This simplification allows us to describe the interaction between the legitimate signal and the perturbation in isolation, and the zero-delay assumption ensures that the autocorrelation peak is always at 0, thus simplifying notation.

Finally, we assume that $x(t)$ and $\Delta x(t)$ are phase-coherent. As illustrated in \autoref{fig:overview}, the perturbation $\Delta x(t)$ is the result of applying a temporal mask or filtering the legitimate signal $x(t)$, rather than modulating an independent carrier. This approach prevents the destructive interference caused by adding passband signals generated using unsynchronized oscillators. While a filter may introduce a deterministic phase rotation to $\tilde{x}(t)$, the ToA estimator neutralizes this rotation by computing the absolute value. We address independent, non-coherent address signals at the end of this section.

\paragraph{ToA estimation}
We define the ToA estimate of the attacked receiver, $\tilde{\tau}_0$, as 
\begin{align}
\tilde{\tau}_0 = \underset{\tau}{\mathrm{argmax}}\,|\tilde{R}(\tau)|^2
\end{align}
where $\Rt(\tau)$ is the cross-correlation
\begin{equation}
\Rt(\tau) = \infint \tilde{x}(t)x^*(t-\tau)\, dt. \label{eq:R_tilde}
\end{equation}
We prefer $|\Rt(\tau)|^2$ over the magnitude $|\Rt(\tau)|$ to simplify the subsequent derivation. By the linearity of the cross-correlation operator, we can write $\Rt(\tau)$ as
\begin{equation}
    \Rt(\tau) = \R(\tau) + \Delta R(\tau)\label{eq:R_linearity},
\end{equation}
where $R(\tau)$ is the autocorrelation of $x(t)$ and $\Delta R(\tau)$ is the cross-correlation between the perturbation $\Delta x(t)$ and the local template. For brevity, we further define $P(\tau) = |R(\tau)|^2$, $\Delta P(\tau) = |\Delta R(\tau)|^2$, and $\tilde{P}(\tau) = |\tilde{R}(\tau)|^2$.

\subsubsection{Finding $\Delta{x}(t)$}
We can characterize the ToA advance by analyzing the first derivative of the cross-correlation. For the autocorrelation $P(\tau)$, the peak is centered at $\tau=0$ by construction and 

\begin{equation}
P'(0)= 2\Re\{R^*(0)R'(0)\} = 0.
\end{equation} 
In a successful attack, the peak of the cross-correlation $\tilde{P}(\tau)$ must be at a $\tilde{\tau}_0 < 0$. Assuming the perturbation $\Delta x(t)$ preserves the unimodality of the cross-correlation in the relevant window around the peak, a sufficient condition for a ToA advance is

\begin{equation}
    \tilde{P}'(0) < 0,\label{eq:derivative_zero_condition}
\end{equation}
As derived in \autoref{sec:appendix_derivation_derivative}, this derivative can be expressed as:
\begin{align}
   \tilde{P}'(0) =-2\cdot \Re\big\{\langle \Delta x, x'\rangle \cdot \langle\tilde{x}, x\rangle^*\big\}\label{eq:derivative_zero}
\end{align}
This result confirms that the perturbation $\Delta x(t)$ must project onto the derivative $x'(t)$, matching the intuition presented in \autoref{sec:intuition}. The inner product $\langle \tilde{x}, x\rangle$ represents the cross-correlation between the attack signal and the local template and factors in the signal's energy. Given that $\Delta x(t)$ is phase-coherent with $x(t)$ and preserves peak unimodality (i.e., it is not strongly anti-correlated with $x(t)$), $\langle \tilde{x}, x\rangle$ acts as a positive real scalar. Consequently, combining \autoref{eq:derivative_zero_condition} with \autoref{eq:derivative_zero} reduces the attack requirement to
\begin{equation}
\Re\left\{\langle \Delta x, x'\rangle \right\} > 0 \Leftrightarrow \Re\left\{\langle \tilde{x}, x'\rangle \right\} > 0,\label{eq:delta_x_req}
\end{equation}
where we use the fact that $\Re\{\langle x, x'\rangle\} = 0$ for signals with finite support, as the autocorrelation derivative at the origin is purely imaginary.

\subsubsection{Estimating the ToA Advance}
We can estimate $\tilde{\tau}_0$ with a first-order Taylor expansion of $\tilde{P}'(0)$ around the origin:
\begin{equation}
\tilde{P}'(\tilde{\tau}_0) = 0\approx \tilde{P}'(0) + \tilde{\tau}_0 \tilde{P}''(0), \end{equation}
which results in
\begin{align}
\tilde{\tau}_0 &\approx -\frac{ \tilde{P}'(0)}{\tilde{P}''(0)}\approx - \frac{\tilde{P}'(0)}{P''(0)}.\label{eq:delta_tau_approx}
\end{align}
This approximation assumes that the perturbation $\tilde{x}(t)$ only minimally alters the shape of the correlation, such that the curvature at the peak is dominated by the original signal ($\tilde{P}''(0)\approx P''(0)$). Since the curvature of an autocorrelation peak is strictly negative, \autoref{eq:delta_tau_approx} shows that any perturbation resulting in $\tilde{P}'(\tau) < 0$ will advance the peak. By substituting \autoref{eq:derivative_zero} into the numerator, we can directly map the attacker's perturbation $\Delta x(t)$ to the resulting time shift. Furthermore, the magnitude of the advance is inversely proportional to $P''(0)$, which is fundamentally tied to the RMS bandwidth of $x(t)$ (see $\beta_{\mathrm{RMS}}$ in \autoref{sec:background}).

\subsubsection{Channel Effects and Phase Alignment}
In the above derivation we assumed zero phase offset between the passband equivalents of $x(t)$ and $\Delta x(t)$. This assumption holds for our attacks as explained at the beginning of this section. However, if the attacker introduces the perturbation $\Delta x(t)$ via their own passband signal, they must ensure the carrier phases of the perturbation and the legitimate signal are aligned. Otherwise, the definition of $\tilde{x}(t)$ changes into
\begin{equation}
    \tilde{x}(t) = x(t) + \Delta x(t)e^{j\phi},
\end{equation}
and the phase shift $\phi$ propagates into the final result:
\begin{equation}
\Re\left\{e^{j\phi}\langle \Delta x, x'\rangle \right\} > 0\label{eq:delta_x_req_phi}
\end{equation}
This phase dependency can have a considerable impact on the attack. If the attacker's carrier is out of phase by $\phi=\pi/2$, the attack has no effect. If the phase difference is $\phi=\pi$, $\Delta x(t)$ causes massive destructive interference with $x(t)$.

Next, we discuss how to construct $\tilde{x}(t)$ dynamically from $x(t)$ without symbol knowledge. 

\blueoff

\subsection{Attack I: Temporal Masking}\label{sec:attack_mixer}
We define the attack signal as $\tilde{x}(t) = x(t)m(t)$, where $m(t)$ is a pre-defined symbol-periodic mask. This attack is suited for \emph{linear modulations employing a fixed pulse-shaping filter} $g(t)$, such as ASK, PSK, or QAM. In these modulations, the signal derivative  is $x'(t) = s \cdot g'(t)$, where $s$ denotes the symbol value unknown to the attacker. While $s$ determines the signal's polarity or phase, the shape of $g'(t)$ is known. Consequently, the attacker can design $m(t)$ to selectively amplify $x(t)$ when $g'(t) > 0$ and attenuate it elsewhere, effectively meeting the requirement derived in \autoref{eq:delta_x_req} and advancing the cross-correlation peak.

\autoref{fig:cut_example} illustrates two attack signals based on this principle: $\tilde{x}_1(t)$ is generated by multiplying $x(t)$, a Gaussian pulse, by a reversed step function, which truncates its trailing half. By contrast, $\tilde{x}_2(t) = x(t)e^{\alpha g'(t)}$ applies a smoother, more subtle distortion (illustrated with $\alpha = 0.3$). Both signals satisfy the condition $\Re\{\langle \Delta x, x'\rangle\} > 0$ and break the symmetry of the cross-correlation by shifting the energy centroid of the waveform to the left. If applied consistently to each symbol of a ranging signal, this bias "pulls" the correlation peak to a negative delay, resulting in a ToA advancement.

\begin{figure}[t]
    \centering
    \includegraphics[width=\linewidth]{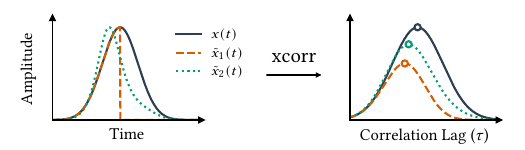}
        \label{fig:cut_example_a}
    \vspace{-1.2em} %
     \caption{Attack signals $\tilde{x}_1(t)$ and $\tilde{x}_2(t)$ (left) cross-correlated with the template $x(t)$ advance the correlation peak (right).}
    \label{fig:cut_example}
\end{figure}

\subsubsection{Frequency-Shift Keying (FSK)}
In FSK, the baseband signal is defined as

\begin{equation}
x(t) = e^{j\phi(t)}\text{, with }\phi(t) = \int_0^t \omega(\tau)\,d\tau
\end{equation}
where $\omega(t)$ denotes the angular velocity. Unlike in linear modulations, using $m(t)$ to reshape the signal does not affect an ideal FSK receiver\footnote{We discuss the impact of non-linear RF front-ends in \autoref{sec:discussion_bt}.}, as the symbol information is carried within the phase trajectory $\phi(t)$. Because the signal has a constant envelope, \autoref{eq:derivative_zero} evaluates to zero (see \autoref{sec:fsk_innerp}), meaning the correlation peak remains stationary.

We observe that $m(t): \mathbb{R}\rightarrow\mathbb{C}$ may alter the phase trajectory of $x(t)$ and influence the correlation peak. However, in this scenario, the ToA shift depends on the unknown symbol value, and for longer sequences, the shifts are normally distributed with mean $\mu = 0$. Thus, achieving a deterministic ToA advancement is not feasible (see \autoref{sec:fsk_tildex}). Furthermore, real-world narrowband systems like Bluetooth CS typically average multiple measurements to enhance accuracy~\cite{BluetoothSpec}, which further mitigates the impact of the attack.

\subsubsection{Practical Considerations}\label{sec:masking_practical_considerations}
The most straightforward implementation of temporal masking involves attenuating or truncating the tail of each symbol, as shown in \autoref{fig:cut_example}. In hardware, this can be realized with an analog RF switch driven by a square wave synchronized to the symbol rate. If the transceivers are within operational range, the attacker can prevent interference from the legitimate signal by amplifying the truncated version using directional antennas or an amplifier. This will trigger the receiver's automatic gain control to adjust to the high power level, effectively attenuating the legitimate signal.

\paragraph{Synchronization} Precise synchronization between the mask and the ranging signal is essential for a successful attack. However, the legitimate receiver must synchronize with the incoming signal as well for demodulation, typically using a predefined preamble. Thus, the attacker can leverage the same mechanisms to acquire synchronization.

\subsection{Attack II: NGD Filtering}\label{sec:attack_filter}
In this section, we describe how a narrowband signal can be advanced in time by passing it through a specific filter.

\subsubsection{Negative Group Delay}
The ideal, non-causal filter that perfectly advances a signal by $\Delta t$ is defined as $h_{\mathrm{ideal}}(t) = \delta(t+\Delta t)$. Our goal is to find a causal, realizable approximation $h(t) \approx h_{\mathrm{ideal}}(t)$, i.e.,

\begin{align}
\tilde{x}(t) = x(t) * h(t) &\overset{\hphantom{\autoref{eq:tildex_taylor}}}{\approx} x(t+\Delta t)\notag \\
&\overset{\autoref{eq:tildex_taylor}}{\approx} x(t) + \Delta t x'(t)+ \mathcal{O}((\Delta t)^2).
\end{align}
where we use the Taylor expansion from \autoref{eq:tildex_taylor}. This approximation requires $h(t)$ to exhibit a \emph{negative group delay} (NGD) within the bandwidth of the input signal. The group delay $\tau_g$ is defined as

\begin{equation}
\tau_g(\omega) \triangleq -\frac{d\phi(\omega)}{d\omega},\label{eq:gdl}
\end{equation}
where $\phi(\omega) = \angle H(\omega)$ is the phase response of the filter. In essence, the group delay characterizes the temporal shift experienced by a signal's envelope when it passes through the system; while a constant positive group delay $\tau_g=\Delta\tau$ results in a time-domain delay $x(t)\rightarrow x(t-\Delta\tau)$, a \emph{constant} NGD causes a non-causal advancement. 

Crucially, a causal and realizable circuit can exhibit NGD within a \emph{restricted frequency band}, provided it maintains a positive group delay elsewhere in the spectrum. If the input signal's energy is concentrated primarily within this NGD region, the output appears advanced with minimal distortion. Thus, our goal is to design a filter $h(t)$ with approximately constant NGD in the band containing the main lobe of $x(t)$, while accepting non-linear and positive group delay at frequencies with negligible spectral energy. As depicted in \autoref{fig:ngd_demo}, such a filter creates the illusion of a time shift by reshaping the signal envelope. Prior work on NGDs exists~\cite{ngd_manitoba, Nako2024, Choietal2010}, but not in the context of secure ranging.

\begin{figure}[t]
  \centering
    \includegraphics[width=\linewidth]{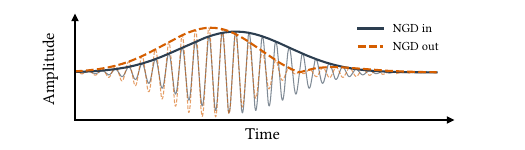}
    \caption{%
    NGD effect on a passband signal: the signal envelope is advanced but suffers a minor distortion of its tail.}
    \label{fig:ngd_demo}
\end{figure}

\subsubsection{Finding $h(t)$} Recalling \autoref{eq:delta_x_req}, we start with a simple filter

\begin{equation}
    \tilde{x}(t) = x(t) * h(t) \approx x(t) + \Delta t x'(t),
\end{equation}
which has the transfer function $H(\omega)$ and the group delay

\begin{align}
     H(\omega)=1+j\omega\Delta t,\qquad\qquad\tau_g(\omega) = -\frac{\Delta t}{1+(\omega\Delta t)^2}
\end{align}
where $\omega=2\pi f$. Although the NGD is promising, the filter is not realizable because 
$|H(\omega)| \rightarrow \infty$ for $\omega \rightarrow \infty$. However, we can further approximate this filter with

\begin{equation}
    H(\omega) = 1 + \frac{\Delta t j\omega}{1+\Delta t j\omega}\Rightarrow\tau_g(\omega) = -\Delta t \frac{1 - 2(\omega \Delta t)^2}{(1 + \omega^2 \Delta t^2)(1 + 4\omega^2 \Delta t^2)}\label{eq:ngd_filter}
\end{equation}

The amplitude response and group delay of this filter are shown in \autoref{fig:ngd_filter} for $\Delta t=\SI{50}{\nano\second}$: the filter has bounded gain and $\tau_g\approx-\SI{50}{\nano\second}$ in the band $0-\SI{0.5}{\mega\hertz}$. Although $H(\omega)$ is a high-pass filter, it only amplifies high-frequency noise by a factor of 2 (\SI{6}{\decibel}) for $\omega\rightarrow \infty$. Thus, this filter advances a \SI{1}{\mega\hertz}-wide baseband signal by roughly \SI{50}{\nano\second} or \SI{15}{\meter} with minor distortion. 

\begin{figure}
 \centering
    \includegraphics[width=\linewidth]{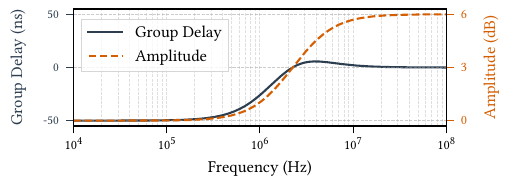}
     \caption{Amplitude response and group delay of $H(f)$ with $f=\omega/(2\pi)$. %
     This filter applies a time-shift of $\Delta t\approx-\SI{50}{\nano\second}$ to signals with $f\leq\SI{1}{\mega\hertz}$.}
       \label{fig:ngd_filter}
\end{figure}

\section{Evaluation: Bluetooth Channel Sounding}\label{sec:evaluation}
In this section, we describe a concrete NGD circuit tailored to Bluetooth CS, and present both a hardware simulation and a physical proof-of-concept (PoC). While the hardware simulation allows for a systematic analysis, the PoC demonstrates the practical feasibility of the attack using commercial off-the-shelf (COTS) components.

\subsection{The NGD Circuit}\label{sec:ngd_circuit}

\begin{figure*}
     \centering
     \begin{subfigure}[b]{0.74\textwidth}
         \centering
         \includegraphics[trim=0cm 20.5cm 0cm 0cm, clip, width=\textwidth]{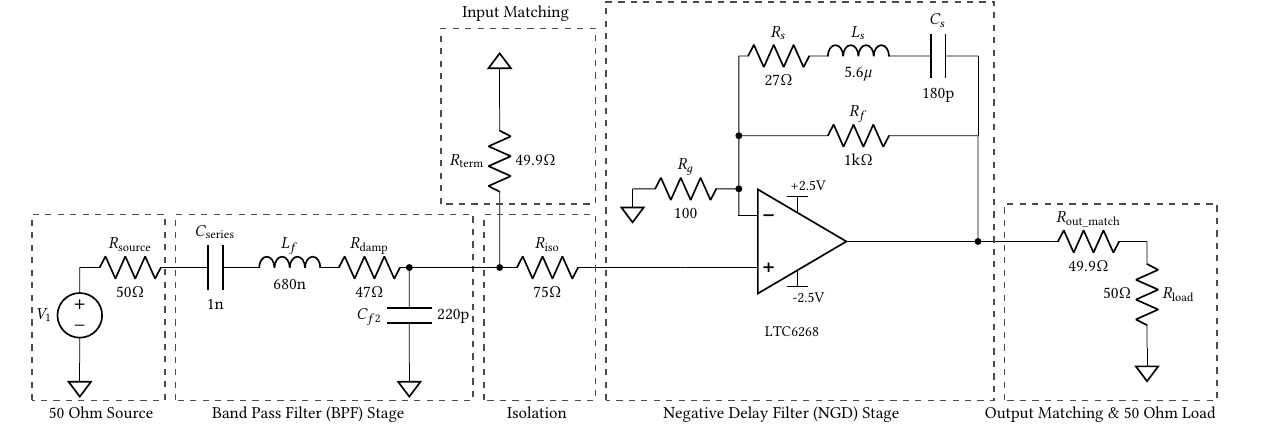} %
         \caption{Circuit schematic}
         \label{fig:ngd_schematic}
     \end{subfigure}
     \hfill
     \begin{subfigure}[b]{0.24\textwidth}
         \centering
    
             \includegraphics[height=4.4cm]{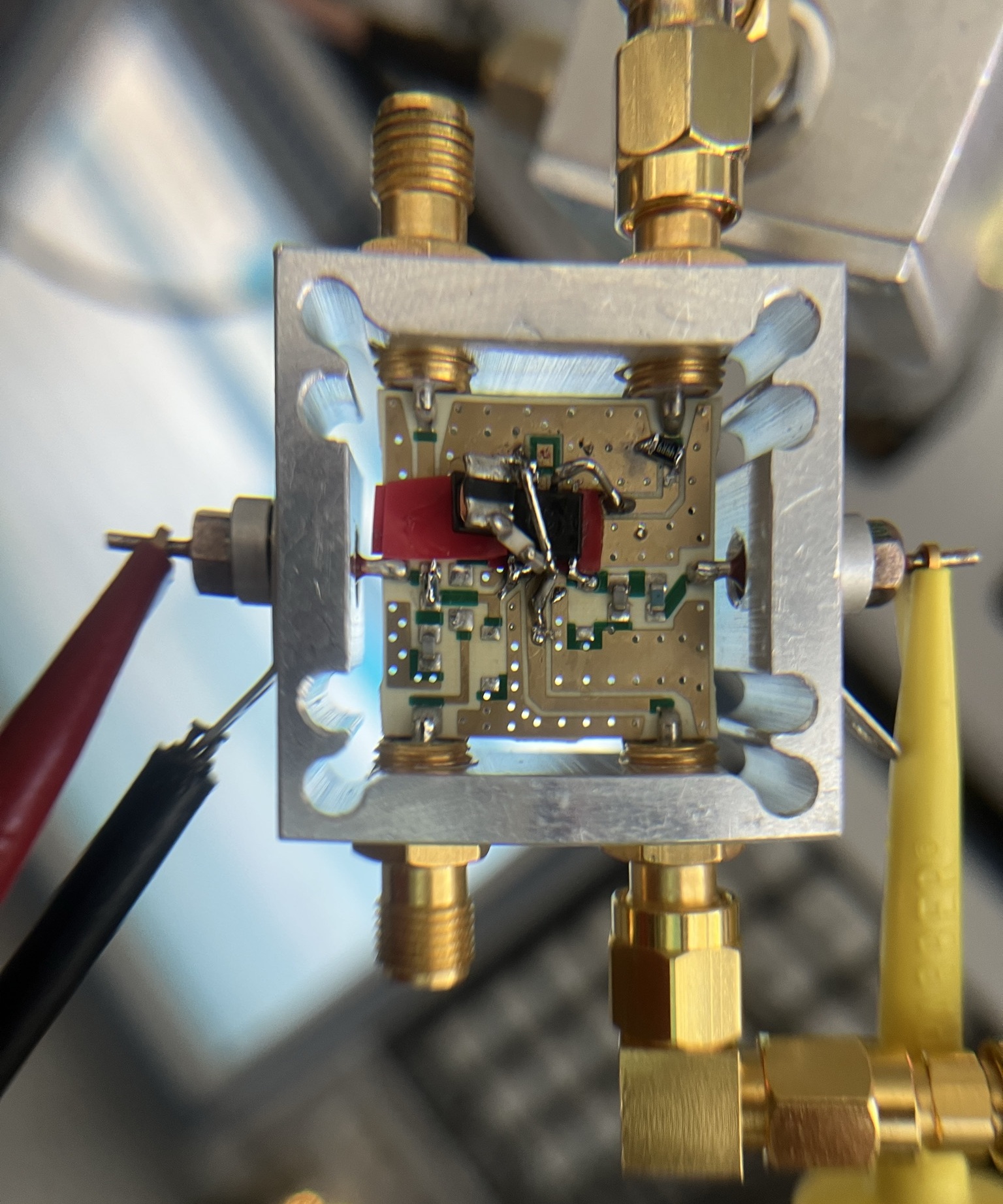}
             \caption{Prototype of the NGD stage}
             \label{fig:poc}
     \end{subfigure}
     
     \caption{NGD circuit based on the LTC6268 operation amplifier.}
     \label{fig:ngd}
\end{figure*}

In practice, NGD circuit operation is more robust at lower intermediate frequencies (IF) than at baseband or at \SI{2.4}{\giga\hertz}. Consequently, our design is optimized for a \SI{4.77}{\mega\hertz} IF, which requires down-converting Bluetooth CS packets to this frequency. We discuss this process in the threat model in \autoref{sec:bluetooth_threatmodel}.

\subsubsection{Design} Several active and passive NGD filter designs have been proposed and implemented in existing literature~\cite{ngd_manitoba}. We choose an active configuration using an operational amplifier to ensure the signal is not excessively attenuated. The circuit schematic is shown in~\autoref{fig:ngd_schematic} and comprises the following stages:
\begin{enumerate}
\item \emph{Input/Output:} A \SI{50}{\ohm} signal source and a \SI{50}{\ohm} load. %
\item \emph{Bandpass filter:} An RLC network to suppress out-of-band noise. This stage adds a positive group delay of around \SI{10}{\nano\second}. %
\item \emph{Isolation and Damping:} Resistors used to stabilize the circuit and fine-tune the frequency response. 
\item \emph{NGD stage:} An active second-order band-stop filter based on the LTC6268 operational amplifier in a non-inverting configuration. %
\item \emph{Impedance Matching:} Input and output impedance are set to \SI{50}{\ohm} with appropriate resistors.

\end{enumerate}

\subsubsection{Hardware Simulation and Prototyping}\label{sec:padding}
We evaluate the circuit design through time-domain and frequency-domain analyses using the Analog Devices LTspice~\cite{ltspice} circuit simulator. To ensure a realistic simulation, we incorporated the following measures:
\begin{enumerate}
    \item \emph{Transient response testing:} The input signal is bracketed by zero-valued padding to test the circuit's stability at packet boundaries.
    \item \emph{High-resolution sampling:} A sampling rate of 200 GSa/s (resolution of \SI{5}{\pico\second}) allows for testing the circuit's stability and parasitic effects over a wide band.
    \item \emph{COTS hardware:} The design only comprises commercially available components. This ensures the model accounts for realistic performance constraints and component-specific non-idealities, such as operational amplifier bandwidth, stability margins, and series resistance of inductors. 
    \item \emph{Modeling of parasitic effects:} We manually added parasitic capacitances to model PCB effects near the negative pin of the operational-amplifier.
    \item \emph{Noise:} To analyze the effect of out-of-channel noise on the circuit, we added Gaussian noise to the input signal in one experiment (see configuration 4 in Experiment 2).
\end{enumerate}

\revlabel{PoC}\blueon  To validate the practical feasibility of the proposed design, we build a physical PoC (\autoref{fig:poc}), which we evaluate in \autoref{sec:experiment_4}.\blueoff

\subsection{Methodology}
Before we present our evaluation in the next section, we summarize key Bluetooth CS features and introduce our threat model.

\subsubsection{Bluetooth CS Primer}\label{sec:bluetooth_background}
The following information is taken from the Bluetooth Core Specification, Version 6.2~\cite{BluetoothSpec}.
\paragraph{Ranging modes} Bluetooth CS supports phase-based (PBR) and ToF/round-trip time (RTT) ranging. Both modes can be used separately or in combination. This work focuses exclusively on RTT ranging.

\paragraph{PHY Layers:} Bluetooth CS provides three different PHYs:  \emph{LE 1M} (\SI{1}{\mega\hertz} channel bandwidth), and \emph{LE 2M}/\emph{LE 2M 2BT} (\SI{2}{\mega\hertz}). We analyze the LE 1M and LE 2M in our experiments.

\paragraph{Frequency-hopping spread spectrum (FHSS)} Bluetooth employs FHSS to mitigate interference and fading by switching channels between packets. The hopping pattern is pseudo-random and determined during the connection procedure.

\paragraph{Packet formats} The RTT mode utilizes the CS SYNC Packet (\autoref{fig:cs_sync}), consisting of a preamble, a 32-bit pseudo-random \emph{Access Address}, and a 4-bit trailer. To improve ToA estimation and security, the packet may optionally include a pseudo-random \emph{Random Sequence} (up to 128 bits) or a \emph{Sounding Sequence} (alternating 1/0 pattern with randomly placed markers to prevent ED/LC attacks).

\begin{figure}
    \centering
    \includegraphics[width=0.9\linewidth]{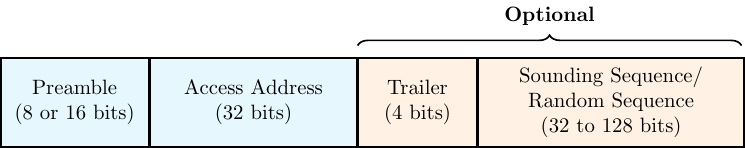}
     \caption{CS SYNC packet in Bluetooth CS RTT ranging.}
    \label{fig:cs_sync}
\end{figure}
\paragraph{Fractional timing estimate (FTE)} The specification details a phase-based algorithm to achieve sub-sample ToA accuracy for Sounding Sequences. For Random Sequences, no algorithm is specified, but receivers are likely to rely on polynomial interpolation.

\paragraph{Cross-Correlation} Bluetooth CS uses \emph{differential cross-correlation}, which provides robustness against carrier frequency offsets. The ToA estimate $\tau_0$ is defined as

\begin{equation} 
\tau_0 = \underset{\tau}{\mathrm{argmax}}\left|\int x(t)x^{*}(t-T_{\mathrm{sym}})s^{*}(t-\tau)s(t-\tau-T_{\mathrm{sym}})\,dt\right|
\label{eq:diff_xcorr}
\end{equation}
where $x(t)$ is the received signal, $s(t)$ is the local reference template, and $T_{\mathrm{sym}}$ the symbol duration. We follow the specification and implement this approach in our receiver.

\paragraph{Normalized Attack Detection Metric (NADM)} 
To flag potential attacks, Bluetooth CS mandates a NADM for secure configurations. While the exact implementation is vendor-specific, the specification suggests three reference metrics, which we refer to as the \emph{Normalized Cross-Correlation (NCC) Metric}, the \emph{Phase Minimum Square Error (PMSE) Metric} and the \emph{DFT Metric}. Our MATLAB receiver computes all three to evaluate the stealthiness of our attack.

\subsubsection{Threat Model}\label{sec:bluetooth_threatmodel}
We adopt the Dolev-Yao model, granting the attacker full control over the wireless channel. This includes the ability to improve the SNR (directional antennas or amplifiers), or to suppress the legitimate signal (obstruction or overshadowing).

We further assume the attacker can perform down-conversion to IF and subsequent up-conversion to the carrier with negligible latency. This was demonstrated in \cite{DBLP:conf/uss/AnlikerCC23} for the  "Mix-Down" attack against UWB. A similar wideband approach can shift all Bluetooth CS channels simultaneously. By utilizing a shared oscillator for both conversions, the attacker ensures no phase offset is introduced. 

Finally, we assume the attacker can synchronize with or bypass the Bluetooth FHSS in the order of microseconds. This could be achieved using a wideband software-defined radio (e.g., USRP X310) equipped with an FPGA-based Polyphase Filter Bank (PFB) and energy detection for rapid channel identification~\cite{nguyen2019energy}. Building such a device is an engineering challenge beyond the scope of this work.

\subsubsection{Experimental Setup}
We simulate the Bluetooth baseband processing in MATLAB. For the transmitter, we use the Mathworks Bluetooth Toolbox~\cite{MathWorksBLEWaveform} to generate CS SYNC packets, exported as complex baseband samples. After the attack, the receiver processes the output signals to estimate the ToA and compute the security metrics listed in \autoref{sec:bluetooth_background}. 

The circuit simulation and the physical PoC are integrated via and controlled by a Python wrapper, which models the RF chains between transmitter, attacker, and receiver.
\paragraph{Simulation} For the circuit simulation, the wrapper upsamples the signal to 200 GSa/s (to emulate an analog waveform) and modulates it to the \SI{4.77}{\mega\hertz} IF for LTspice processing. The result then passes through a modeled RF receiver chain consisting of a 4th-order Butterworth bandpass filter and down-conversion to baseband. Finally, the wrapper decimates the signal to \SI{8}{\mega\hertz} to emulate the analog-to-digial converter (ADC) stage. The output is then processed by the MATLAB receiver. 
\paragraph{Prototype} For the PoC implementation, the wrapper feeds baseband samples into a Rohde \& Schwarz SMW200A Vector Signal Generator, which generates the analog passband signal at a \SI{5.379}{\mega\hertz} IF\footnote{The physical circuit transfer function differs slightly from the simulation, thus changing the IF.}. The signal is routed via an SMA cable to the NGD prototype (see \autoref{fig:poc}), and the resulting output is captured by a Rohde \& Schwarz RTP oscilloscope. The waveform is then returned to the wrapper for down-conversion, low-pass filtering, and decimation to \SI{8}{\mega\hertz}.

\paragraph{The Bluetooth Receiver} 
\revlabel{1/5}\blueon
For each received packet, our MATLAB receiver produces a ToA estimate and NADM metrics in four stages:
\begin{enumerate}
    \item \emph{Cross-correlation:} The receiver constructs the correlation template and computes the cross-correlation using \autoref{eq:diff_xcorr}.
    \item \emph{Bit correctness:} The receiver validates the decoded CS SYNC (Access Address and optional sequence), aborting if any bit errors are present.
    \item \emph{ToA estimation:} The receiver computes the FTE depending on the optimal sequence used.
    \item \emph{NADM metrics:} Finally, the receiver computes the detection metrics.
\end{enumerate}
\blueoff

To isolate the ToA advance caused by the attack, we calculate the ToA difference between each processed packet and a ground-truth. This reference signal traverses the same simulation, including frequency conversion and the receiver-side RF chain, but bypasses the NGD circuit. This ensures that any ToA difference detected by the receiver is the result of the attack.

\subsection{Experiments}\label{sec:experiments}
We conducted four experiments to (i) analyze the circuit's time- and frequency-domain behavior, (ii) estimate the ToA advance across various Bluetooth CS configurations, (iii) assess the attack's impact on detection metrics, and (iv) validate practical feasibility using a hardware PoC.

\subsubsection{Experiment 1: Testing the NGD Circuit (Simulation)}\label{sec:experiment_1}
In this first experiment, we evaluate the time- and frequency-domain characteristics of our NGD circuit.

\paragraph{Parameters} We pass through the circuit a single noise-less CS SYNC packet carrying a 128-bit Random Sequence on the LE 1M PHY.

\paragraph{Results:} The frequency response of the NGD circuit is illustrated in \autoref{fig:ngd-response}.
At the \SI{4.77}{\mega\hertz} IF, the circuit achieves a group delay of \SI{-62}{\nano\second}. The NGD band is sufficiently wide for a \SI{1}{\mega\hertz} Bluetooth CS channel and the NGD is nearly constant, ensuring minimal distortion. Due to the active filter design, the attenuation is limited to \SI{10}{\decibel}. Furthermore, the integrated bandpass filter suppresses out-of-band frequencies, with amplitude peaks below \SI{5}{\decibel} and \SI{0}{\decibel}, respectively.

The comparison of the cross-correlation peaks of input and output (see \autoref{fig:hw_sim_corr}) shows an advance of \SI{-61.35}{\nano\second}, matching our expectations. Finally, \autoref{fig:hw_sim_time} shows the first symbols of the CS SYNC packet and demonstrates that the distortion is negligible. The signal advancement is most clearly visible at the troughs.

\begin{figure}
    \centering
    \includegraphics[width=\linewidth]{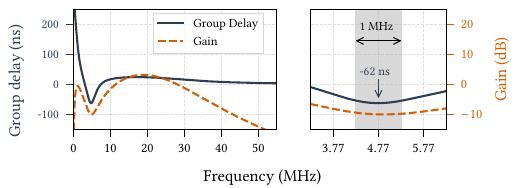}
    \caption{Frequency response of the simulated NGD circuit.}
    \label{fig:ngd-response}
\end{figure}

\begin{figure}
    \centering
    \includegraphics[width=\linewidth]{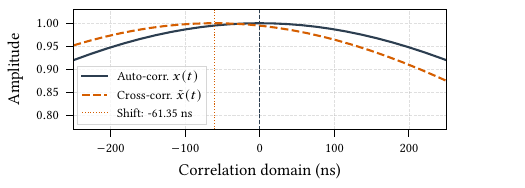}
    \caption{Correlation of the NGD output $\tilde{x}(t)$ with the template $x(t)$, peaking -61.35 ns earlier than the autocorrelation.}%
    \label{fig:hw_sim_corr}
    \end{figure}
    
\begin{figure}
    \centering
        \includegraphics[width=\linewidth]{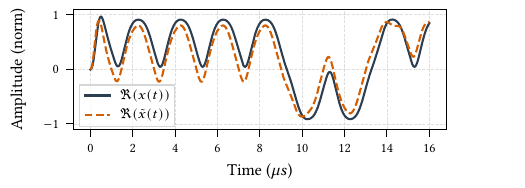}
        \caption{Input $x(t)$ and the output $\tilde{x}(t)$ of the NGD stage. Peaks and troughs of Re($\tilde{x}(t)$) appear slightly left-shifted.} 
    \label{fig:hw_sim_time}
\end{figure}

\subsubsection{Experiment 2: PHY Comparison (Simulation)}\label{sec:experiment_2}
\paragraph{Parameters} We generated 100 CS SYNC packets for each of the following configurations:
\begin{enumerate}
    \item LE 1M PHY, 128-bit Random Sequence
    \item LE 1M PHY, 96-bit Sounding Sequence
    \item LE 2M PHY, 128-bit Random Sequence
    \item LE 1M PHY, 128-bit Random Sequence, input SNR=\SI{10}{\decibel} 
\end{enumerate}

The first three configurations involve noise-free signals. In the fourth, we increased the sampling rate from 8 MSa/s to 80 MSa/s and injected per-sample AWGN to simulate high-frequency noise. This configuration models a realistic scenario where the attacker receives a noisy signal.

\paragraph{Results} \autoref{fig:distance_reduction} illustrates the ToA advancement of the attacked packets relative to the ground truth. The results are in line with our expectations (compare \autoref{fig:ngd-response}). Packets using the LE2M PHY experience a slightly shorter advancement ($\mu=\SI{-16.96}{\meter}, \sigma=\SI{4.25}{\centi\meter}$) than their \SI{1}{\mega\hertz} counterparts ($\mu=\SI{-18.35}{\meter}, \sigma=\SI{2.83}{\centi\meter}$), which we attribute to the bandwidth difference; more spectral energy falls on the rising edges of the NGD trough, increasing distortion and reducing the effective average NGD. 

Sounding Sequence packets are similarly affected, with a slightly higher mean ($\mu=\SI{-19.09}{\meter}$). The variance is lower ($\sigma=\SI{0.09}{\centi\meter}$), since Sounding Sequences are mostly deterministic, where as Random Sequences are pseudo-random. 
Finally, the last configuration demonstrates that the attack works under noisy conditions as well ($\mu=\SI{-18.36}{\meter}, \sigma=\SI{6.06}{\centi\meter}$). This is an advantage over ED/LC, which typically relies on a high SNR for the early-detection stage.  %

\begin{figure}
  \centering
    \includegraphics[width=\linewidth]{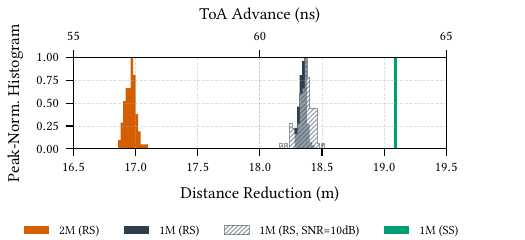}
    \caption{ToA advance by the NGD filter: 2 MHz ($\mathrm{2M}$) signals are advanced less than their 1 MHz counterparts.} %

    \label{fig:distance_reduction}
\end{figure}

\begin{figure*}
  \centering
   \begin{subfigure}[b]{\linewidth}
    \includegraphics[width=\linewidth]{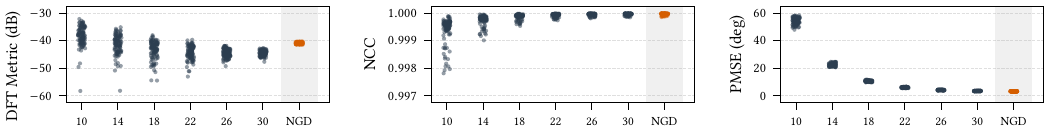}
    \caption{LE1M PHY, 128-bit Random Sequence ($n=100$)}
        \label{fig:experiments_details_a}
    \end{subfigure}
   \begin{subfigure}[b]{\linewidth}
    \includegraphics[width=\linewidth]{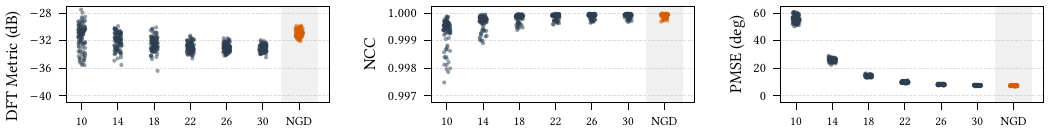}
    \caption{LE2M PHY, 128-bit Random Sequence ($n=100$)}
        \label{fig:experiments_details_b}
    \end{subfigure}
      \begin{subfigure}[b]{\linewidth}
    \includegraphics[width=\linewidth]{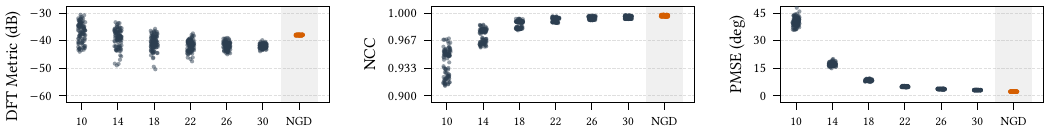}
    \caption{LE1M PHY, 96-bit Sounding Sequence ($n=100$)}
        \label{fig:experiments_details_c}
    \end{subfigure}
    \caption{Impact of NGD filtering on Bluetooth CS detection metrics for different packet configurations, compared with legitimate AWGN (SNR between 10 and 30 dB). The results show that detecting the attack with existing tools is difficult.}
    \label{fig:experiments_details}
\end{figure*}

\subsubsection{Experiment 3: Detection Metrics (Simulation)}
We implemented the three metrics mentioned in \autoref{sec:bluetooth_background} to evaluate their effectiveness in detecting the attack. In particular, we compare the attack's impact on these metrics against the effect caused by typical channel noise.

\paragraph{Parameters} We simulate 100 CS SYNC packets in the following three configurations:
\begin{enumerate}
    \item LE 1M, 128-bit Random Sequence
    \item LE 2M, 128-bit Random Sequence
    \item LE 1M, 96-bit Sounding Sequence
\end{enumerate}
For the ground-truth packets, we set the SNR to values in the range of \SI{10}{}-\SI{30}{\decibel}.

\paragraph{Results} \autoref{fig:experiments_details} shows that the metric distributions under noise and under attack are so similar that applying strict thresholds inevitably results in high false-rejection rates. This matches our expectation, since the distortion caused by the circuit is subtle (see \autoref{fig:hw_sim_time}). We conclude that the attack cannot be detected with existing metrics.

\subsubsection{Experiment 4: Physical PoC}\label{sec:experiment_4}
\revlabel{PoC}\blueon  In this experiment, we validate our physical proof-of-concept NGD implementation.

\paragraph{Parameters:} We generated 100 CS SYNC packets using the same configurations as in Experiment 3. The packets were converted to analog passband signals using a signal generator and fed into the circuit.
\paragraph{Results} The results are shown in \autoref{fig:distance_reduction_hw} in the appendix and confirm that the concept of a low-cost NGD circuit is practically feasible, with considerable reductions generally in line with our expectations. The PoC circuit consists of slightly different components than the simulation (e.g., resistors models and tolerances) and is affected by parasitic effects (e.g., parasitic capacitances) specific to the implementation, which explains slightly worse performance on the \SI{2}{\mega\hertz} signals. In addition, the PoC circuit does not include the bandpass filter stage, which explains the slightly larger distance reductions (around \SI{20}{\meter}) for 1 MHz packets compared to \autoref{fig:distance_reduction}.

\blueoff
\section{Discussion}

\paragraph{Practical constraints and impairments} The above results demonstrate the efficacy of NGD filtering. Although our evaluation utilizes a MATLAB-based receiver and simulated RF chains, the similarity between the filter's input and output waveforms suggests that the attack is feasible against commercial transceivers as well. Common hardware impairments, such as clock frequency offsets, sampling frequency offsets (SFO), and frontend non-linearities are unlikely to mitigate the attack. Furthermore, dynamic channel conditions do not counteract the deterministic NGD introduced by the circuit. Finally, Frequency Hopping Spread Spectrum (FHSS) presents a practical challenge, but it is a functional barrier and not a security primitive that offers long-term protection.

\paragraph{Attack comparison}
The masking attack is relatively straightforward to implement but generally ineffective against constant-envelope or multi-carrier modulation schemes, such as FSK or OFDM (although non-idealities in receivers introduce vulnerabilities, as discussed in \autoref{sec:discussion_bt}). Furthermore, the attack must be tailored to a specific ranging system, as it depends on specific parameters such as symbol duration, pulse shapes, or packet lengths. 
By contrast, the NGD attack requires a more sophisticated circuit but, once implemented, it is protocol-agnostic and universally applicable to any ToF ranging system operating within its design bandwidth.

\paragraph{Attacks published by Bluetooth}\label{sec:discussion_bt}
In November 2025, the Bluetooth Special Interest Group (SIG) disclosed "amplification-based attacks" against Bluetooth CS~\cite{bluetooth2025core62featureoverview}, which rely on symbol-periodic amplification to bias ToA estimates. According to the disclosure, this attack exploits "amplitude-to-phase conversion in the receiver", suggesting that amplification can lead to gain compression in a receiver's front-end, causing a phase shift that impacts the ToA estimate. 

In terms of signal alterations, these attacks resemble the temporal masking presented in this paper, mainly differing in the specific masks applied. Functionally, however, the attacks are fundamentally different: temporal masking targets the cross-correlation, whereas amplification-based attacks exploit hardware non-linearities.

To mitigate the attack, the Bluetooth SIG introduced the "DFT metric," which leverages the constant envelope of GFSK signals to effectively detect periodic amplitude anomalies.

During our research, we further found that amplification attacks also affect the fractional timing estimate (FTE) using Sounding Sequences by breaking the conjugate phase symmetry in the frequency-domain. However, due to the efficacy of the DFT metric, we did not investigate this further.

\subsection{Countermeasures}
\paragraph{Temporal masking} The DFT metric compares a signal's energy at the symbol frequency relative to the DC component. Since GFSK signals have a constant complex envelope, periodic amplitude deviations introduced by temporal masking are easily detectable.

However, for signals with inherent amplitude variations, due to pulse shaping or a non-constant envelope, the metric yields high baseline values by default and is likely ineffective for attack detection.

A possible direction for future work is to measure the accumulation rate of the correlation's inner product. Because the attack distorts the symbol's shape, energy aggregation happens disproportionally early in the symbol. However, the robustness of such a metric under noise is unclear, as it potentially suffers from false positives under noise.

\paragraph{NGD filtering} Our results indicate that a carefully engineered NGD circuit introduces minimal distortion (see \autoref{fig:hw_sim_time}), which is difficult to detect with existing countermeasures (see \autoref{fig:experiments_details}). Fundamentally, an NGD filter exploits the predictable nature of narrowband signals, and breaking the predictability without increasing the bandwidth is hard. For example, sharp symbol transitions or frequency stitching (i.e., combining neighboring channels into a single wider channel) would mitigate the issue, but degrade the system's spectral efficiency.

\paragraph{Phase-based Ranging (PBR)} PBR estimates the distance between devices 
based on the phase shift of continuous-wave tones across multiple frequencies, yielding high-precision measurements. However, combining ToF ranging with PBR does not reliably secure Bluetooth CS against the attacks proposed in this work. First, PBR has been independently shown to be vulnerable to distance-decreasing attacks~\cite{DBLP:conf/ches/OlafsdottirRC17,staat2022analog}. Second, its reliance on narrowband tones renders PBR also susceptible to the NGD attack.

\section{Related Work}
Clulow et al.~\cite{DBLP:conf/esas/ClulowHKM06} first identified the security risks of inferring transmitted bits from partial observations in secure ranging. This insight laid the groundwork for ED/LC attacks across diverse modulation schemes, such as ISO1443 Radio Frequency Identifier (RFID)~\cite{DBLP:conf/wisec/HanckeK08}, Chirp Spread Spectrum (CSS)~\cite{DBLP:conf/wisec/RanganathanDFC12}, Impulse Radio Ultra-Wideband (IR-UWB)~\cite{DBLP:conf/wisec/FluryPPHB10}, multi-carrier ranging~\cite{leu2021security}, and direct-sequence spread spectrum signals as used in GNSS~\cite{zhang2019effects}. Consequently, the IEEE 802.11az task group (Next Generation Positioning) and the Bluetooth SIG analyzed these threats and incorporated countermeasures into their respective standards~\cite{Batra2020_IEEE80211az, BluetoothSpec}.

Recently, a lot work in the field of secure ranging was focused on UWB. Singh et al. showed that ToA estimation in non-line-of-sight scenarios is is potentially vulnerable against adversarial noise injection~\cite{singh2021security}, an idea initially proposed by Poturalski et al.~\cite{poturalski2010cicada}. Leu et al. demonstrated the attack practically~\cite{DBLP:conf/uss/LeuCHRAHCC22}, Joo et al. proposed detection mechanisms~\cite{joo2024enhancing}, and Luo et al. eventually demonstrated how to secure such systems~\cite{luo2024secure}. Further work analyzed the risk of clock drift compensation~\cite{DBLP:conf/uss/AnlikerCC23}, attacks on time-difference-of-arrival ranging~\cite{stocker2020towards}, or how to test chip security in practice~\cite{aad2026fast}.

Phase-based ranging (PBR) has been analyzed by Ólafsdottir et al., who pointed out fundamental security issues~\cite{DBLP:conf/ches/OlafsdottirRC17}, and Staat et al.~\cite{staat2022analog}, who demonstrated attacks against Bluetooth PBR. Although a lot of research on Bluetooth security exists (e.g., \cite{antonioli2022blurtooth, antonioli2019low, lonzetta2018security}), this paper is, to the best of our knowledge, the first to analyze the PHY layer of RTT ranging in Bluetooth CS.

NGD filters have been designed and implemented in specialized literature~\cite{ngd_manitoba, Nako2024}, but their implications for secure ranging and distance bounding remain largely unexplored.

\section{Conclusion}
We analyzed the security of ToA estimation under narrowband conditions and presented two new distance-decreasing attacks against ToF ranging, namely temporal masking and NGD filtering. The attacks do not require real-time symbol detection but instead exploit the narrowband nature of waveforms. While temporal masking is straightforward to implement, it only applies to linear modulations, such as ASK, PSK, or QAM. By contrast, the NGD approach requires more sophisticated hardware but is largely modulation-agnostic. We simulated such a circuit and evaluated the NGD attack against Bluetooth CS, demonstrating that it can advance ToA estimates by up to \SI{60}{\nano\second}, corresponding to distance reductions of \SI{18}{\meter}. Finally, we showed that the attack is difficult to detect with existing metrics.

\begin{acks}
The authors would like to thank Harshad Sathaye for his invaluable technical support in constructing the NGD circuit. We also extend our gratitude to Joaquín Castañon and Jason Zibung, whose work inspired this project. This research has received funding from the Swiss National Science Foundation under NCCR Automation, grant agreement 51NF40\_180545 and from the Zurich Information Security Center (ZISC).
\end{acks}

\bibliographystyle{ACM-Reference-Format}
\bibliography{bibliography}

\appendix
\section{Derivation of $\tilde{P}'(0)$}\label{sec:appendix_derivation_derivative}
\revlabel{4}\blueon
We start by rewriting $\tilde{P}(\tau) = |\tilde{R}(\tau)|^2$ as
\begin{align}
   \tilde{P}(\tau) &= |R(\tau)|^2 + 2\Re\{R(\tau)\Delta R^*(\tau)\} + |\Delta R(\tau)|^2\notag \\
   &= \Re\{R(\tau) R^*(\tau) + 2R(\tau)\Delta R^*(\tau) + \Delta R(\tau)\Delta R^*(\tau)\}
\end{align}
and compute the derivative of the first term as

\begin{align}
\frac{d}{d\tau}\Big[\Re\{R(\tau) R^*(\tau)\}\Big]  &= \Re\left\{\frac{d}{d\tau}\Big[R(\tau)R^*(\tau)\Big]\right\} \notag \\
&\overset{(a)}{=} \Re\left\{R'(\tau)R^*(\tau) + R(\tau)(R^*(\tau))'\right\} \notag \\
&\overset{(b)}{=} \Re\left\{R'(\tau)R^*(\tau) + R(\tau) R'^*(\tau)\right\} \notag \\
&\overset{(c)}{=} 2\Re\{R^*(\tau)R'(\tau)\}
\end{align}
where in (a) we apply the product rule, in (b) we use $(R^*)' = (R')^*$, and in (c) the fact that the two terms are complex conjugates. Proceeding analogously for the other terms, we obtain

\begin{align}
    \tilde{P}' &= 2\Re\left\{R^*R' + R'^*\Delta R + R^*\Delta R' +\Delta R^* \Delta R' \right\}
\end{align}
where we omitted $\tau$ for brevity. To simplify this expression, we evaluate it at $\tau = 0$. The steps are explained in detail below.

\begin{align}
    \tilde{P}'(0) &= 2\Re\left\{R^*(0)R'(0) + R'^*(0)\Delta R(0) + R^*(0)\Delta R'(0) +\Delta R^*(0) \Delta R'(0) \right\} \notag \\
    &\overset{(a)}{=}2\Re\left\{R^*(0)\Delta R'(0) +\Delta R^*(0) \Delta R'(0) \right\} \notag \\
    &\overset{(b)}{=}2\Re\left\{\langle x, x\rangle\cdot\Delta R'(0) + \langle x, \Delta x\rangle \cdot \Delta R'(0) \right\} \notag \\
    &\overset{(c)}{=}2\Re\left\{\langle x, x\rangle\cdot (-\langle \Delta x, x'\rangle) + \langle x, \Delta x\rangle \cdot (-\langle \Delta x, x'\rangle)\right\} \notag \\
    &\overset{(d)}{=}-2\Re\left\{\langle \Delta x, x'\rangle \cdot (\langle x, x\rangle\ + \langle x, \Delta x\rangle) \right\} \notag \\
    &\overset{(e)}{=}-2\Re\big\{\langle \Delta x, x'\rangle\cdot \langle\tilde{x}, x\rangle^*\big\} \notag \\
\end{align}

In (a), we used $\Re\{R'(0)\} = \Re\{R'^*(0)\} = 0$, since the derivative of the autocorrelation $R$ at the peak $\tau = 0$ is imaginary. This cancels the first term and, because we assume $\Delta R(0) \in \mathbb{R}$, the second.

In (b), we used the definition of $R(\tau)$ and $\Delta R(\tau)$ according to \autoref{eq:R_tilde} and \autoref{eq:R_linearity}. For $\tau=0$, the respective signals are aligned in $\tau$, justifying the inner product notation. Furthermore, we used the identity $\langle \Delta x, x\rangle^* = \langle x, \Delta x \rangle$.

In (c) we substituted the derivative of $\Delta R'(\tau)$, which is defined as:

\begin{align}
    \Delta R'(\tau) &= \frac{d}{d\tau} \left[\int_{-\infty}^{\infty} \Delta x(t) x^*(t-\tau) \,dt\right] \notag \\
    &= -\int_{-\infty}^{\infty} \Delta x(t) (x')^*(t-\tau)\, dt.\notag \\
\end{align}
Evaluated at $\tau = 0$, this results in $\Delta R'(0) = -\langle \Delta x, x' \rangle$

In (d), we performed simple algebraic manipulations and in (e), we used $\tilde{x}(t) = x(t) + \Delta x(t)$.
\blueoff

\subsection{Masking Attack against FSK}
\subsubsection{Real Mask}\label{sec:fsk_innerp}
Let $x(t) = e^{j\phi(t)}$ and $\tilde{x}(t) = m(t)e^{j\phi(t)}$. If we consider $m(t): \mathbb{R}\rightarrow\mathbb{R}$, the derivative $x'(t)$ is defined as

\begin{equation}
    x'(t) = j\phi'(t)e^{j\phi(t)},
\end{equation}
and the inner product $\langle\tilde{x},x'\rangle$ as:
\begin{align}
\langle \tilde{x}, x' \rangle &= \infint \tilde{x}(t) \left( x'(t) \right)^* dt \notag \\
&= \infint \left( m(t)e^{j\phi(t)} \right) \left( -j\phi'(t)e^{-j\phi(t)} \right) dt \notag \\
&\overset{(a)}{=} -j \cdot \underbrace{\infint m(t)\phi'(t) dt}_{\in \mathbb{R}} \\
&\Rightarrow -2\Re\Big\{\langle \tilde{x}, x' \rangle\Big\} = 0
\end{align}
In (a), the complex-conjugate phase rotations cancel. Since the integral evaluates to a real number, the final product is strictly imaginary.

\subsubsection{Complex Mask}\label{sec:fsk_tildex}
\begin{lemma}
    Let $x(t) = e^{j\phi(t)}$. There is no $\tilde{x}(t) = m(t)x(t)$ with $m: \mathbb{R} \rightarrow \mathbb{C}$ such that

    \begin{equation}
        \Re\big\{\langle \tilde{x},x'\rangle\big\} > 0, \forall x
    \end{equation}

    for all $\phi(t): \mathbb{R} \rightarrow\mathbb{R}$. Consequently, no universal mask exists that can successfully advance the ToA for all possible FSK symbol sequences.
    
\end{lemma}

\paragraph{Proof:} For $m(t): \mathbb{R} \rightarrow\mathbb{R}$, the result has been shown in \autoref{sec:fsk_innerp}. For $m(t): \mathbb{R} \rightarrow\mathbb{C}$, the inner product $\langle \tilde{x}, x' \rangle$ yields:

\begin{align}
 \langle \tilde{x}, x' \rangle &= \infint \left( m(t)e^{j\phi(t)} \right) \left( -j\phi'(t)e^{-j\phi(t)} \right)\,  dt \notag \\
&= \infint -j \left( m_R(t) + jm_I(t) \right) \phi'(t)\,  dt  \notag\\
&= \infint \big(m_I(t) -jm_R(t)\big) \phi'(t)\, dt \notag\\
\Rightarrow \Re\big\{\langle \tilde{x}, x' \rangle\big\} &= \infint m_I(t)\phi'(t) dt
\end{align}

Assuming $m_I$ exists such that $\Re\big\{\langle \tilde{x}, x' \rangle\big\} > 0$ then $x_2(t) = e^{-j\phi(t)}$ (effectively flipping the bits) yields $\Re\big\{\langle \tilde{x}, x_{2}' \rangle\big\} < 0$, thus proving the lemma.

\section{Results of the Prototype}
The results of the prototype are illustrated in~\autoref{fig:distance_reduction_hw}:

\begin{figure}[h]
  \centering
    \includegraphics[width=\linewidth]{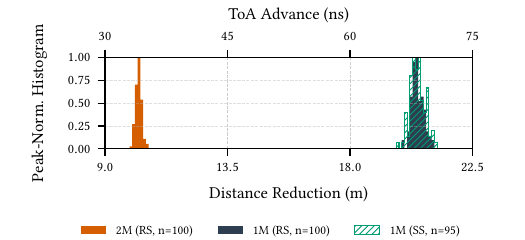}
    \caption{ToA advance measured after passing CS SYNC packets through a physical NGD circuit. Note that this circuit lacks the bandpass filter stage, which would incur a group delay penalty of $\approx\SI{10}{\nano\second}$.}
    \label{fig:distance_reduction_hw}
\end{figure}

\end{document}